%% file: paper.tex
\documentclass[amsthm]{elsart}

\usepackage{yjsco}
\usepackage{mathptmx}
\usepackage{epsfig}
\usepackage{amsthm}
\usepackage{amsmath}
\usepackage{mathbbol}
\usepackage{amssymb}
\usepackage{graphicx}
\usepackage{algorithmic}
\usepackage{url}
\usepackage{epsfig}
\usepackage{wrapfig}
\usepackage{algorithm}

\input papermacros

\usepackage{url} 
\urldef{\mails}\path|msagralo@mpi-inf.mpg.de|

\newcommand{\var}{\operatorname{var}}
\newcommand{\Otilde}{\tilde{O}}

\newcommand{\dcn}{\textsc{Dsc}^2}

\numberwithin{equation}{section}
\numberwithin{figure}{section}

\begin{document}

\begin{frontmatter}

\title{When Newton meets Descartes: A Simple and Fast Algorithm to Isolate the Real Roots of a Polynomial\\
}

%\thanks{This research was partly supported by .....}

\author{Michael Sagraloff}
\address{MPI for Informatics, Saarbr\"ucken, Germany}
\ead{msagralo@mpi-inf.mpg.de}
%\ead[url]{URL 2}

\begin{abstract}
We introduce a new algorithm denoted $\dcn$ to isolate the real roots of a univariate square-free polynomial $f$ with integer coefficients. The algorithm iteratively subdivides an initial interval which is known to contain all real roots of $f$. The main novelty of our approach is that we combine Descartes' Rule of Signs and Newton 
iteration. More precisely, instead of using a fixed subdivision strategy such as bisection in 
each iteration, a Newton step based on the number of sign variations 
for an actual interval is considered, and, only if the Newton step fails, we fall back to bisection. Following this approach, our analysis shows that, for most iterations, we can achieve quadratic convergence towards the real roots. In terms of complexity, our method induces a recursion tree of almost optimal size $O(n\cdot\log(n\tau))$, where $n$ denotes the degree of the polynomial and $\tau$ the bitsize of its coefficients. The latter bound constitutes an improvement by a factor of $\tau$ upon all existing subdivision methods for the task of isolating the real roots. In 
addition, we provide a bit complexity analysis showing that $\dcn$ needs only $\tilde{O}(n^3\tau)$ bit operations\footnote{$\tilde{O}$ indicates that we omit logarithmic factors.} to
isolate all real roots of $f$. This matches the best bound known for this fundamental problem. However, in 
comparison to the much more involved algorithms by Pan and Sch\"onhage (for the task of isolating all complex roots) which achieve the same bit 
complexity, $\dcn$ focuses on real root isolation, is very easy to access and easy to implement.
\end{abstract}


\end{frontmatter}

\section{Introduction}\label{intro}

Finding the roots of a univariate polynomial $f$ is considered as one of the most important tasks in computational algebra. This is justified by the fact that many problems from mathematics, engineering, computer science, and the natural sciences can be reduced to solving a system of polynomial equations which in turn, by means of elimination techniques such as resultants or Gr\"obner Bases, reduces to solving a polynomial equation in one variable. Hence, it is not surprising that numerous approaches are dedicated to this fundamental problem.
We mainly distinguish between (1) numerical and exact methods, and (2) methods to find all complex roots and methods which are especially tuned to search for real roots. The numerical literature lists many algorithms, such as  Newton iteration or the Weierstrass-Durand-Kerner method, that are widely used and effective in practice but lack a guarantee on the global behavior (cf. \cite{pan:history-progress:97} for discussion). In particular, the convergence and/or the complexity of the Weierstrass-Durand-Kerner method is still open.

The work of A.\,Sch\"onhage \cite{schonhage:fundamental} from 1982 marks the beginning of the complexity-theoretic approaches. It combines a newly introduced concept denoted \emph{splitting circle method} with techniques from numerical analysis (Newton iteration, Graeffe's method, discrete Fourier transforms) and fast algorithms for polynomial and integer multiplication. For the benchmark problem of isolating all complex roots of a polynomial $f$ of degree $n$ with integer coefficients of modulus $2^\tau$ or less, the proposed method achieves the record bound of $\tilde{O}(n^3\tau)$ bit operations. V.\,Pan and others~\cite{pan:seq-parallel:87,pan:history-progress:97} gave theoretical improvements in the sense of achieving record bounds simultaneously
in both bit complexity and arithmetic complexity, but the initial bound $\tilde{O}(n^3\tau)$ on the number of bit operations has still remained intact. Common to all asymptotically fast algorithms is (as the authors themselves admit) that they are rather involved and very difficult to implement. The latter is also due to the fact that, in order to control the precision errors in the considered numerical subroutines, one has to carefully work out many details of their implementation. Hence, it is not surprising that, despite their theoretical richness, the asymptotically fast algorithms have so far not been used, or not proven to be efficient in practice; see~\cite{Gourdon96} for an implementation of the splitting circle method within the Computer Algebra system Pari/GP. A further reason might be that the benchmark problem is inappropriate for most applications. For instance, in ray shooting in computer graphics, we are only interested in the first positive root or in the real roots in some specified neighborhood. 

In parallel to the development of purely numerical methods, there is a steady ongoing research on exact subdivision algorithms, such as the Descartes method~\cite{collins-akritas:76,eigenwillig:thesis,krandick-mehlhorn:06,mrr:bernstein:05,rouillier-zimmermann:roots:04}, the Bolzano method~\cite{DBLP:journals/corr/abs-1102-5266,DBLP:journals/eccc/BurrKY09,burr-sharma-yap:eval:09,mitchell:robust-ray:90,DBLP:conf/issac/YapS11}, Sturm Sequences~\cite{du-sharma-yap:sturm:07,lickteig-roy:sequences:01,reischert:subresultant:97} or the continued fraction method~\cite{akritas-strzebonski:comparison:05,sharma:continued-frac:08,te-cf:08}. These methods from the exact computation literature are widely used in various algebraic applications (e.g., cylindrical algebraic decomposition), and many of them have been integrated into computer algebra systems (e.g., \textsc{Maple}, \textsc{Mathematica}, \textsc{Sage}, etc.). In addition, their computational complexity has been well-studied~\cite{du-sharma-yap:sturm:07,eigenwillig-sharma-yap:descartes:06,sharma:continued-frac:08,te-cf:08,DBLP:conf/issac/YapS11}, and many experiments have shown their practical evidence~\cite{snc-benchmarks09,johnson:root-isolation:98,rouillier-zimmermann:roots:04}. Current experimental data shows that a version of the Descartes method (i.e., the univariate solver in \textsc{Rs} based on~\cite{rouillier-zimmermann:roots:04}, integrated into \textsc{Maple}) which uses approximate computation performs best for most polynomials, whereas, for harder instances, the continued fraction approach seems to be more efficient. 
With respect to the benchmark problem, all of the above mentioned algorithms demand for $\tilde{O}(n^4\tau^2)$ bit operations to isolate all real roots, hence they tend to lag behind the asymptotically fast algorithms by a factor of $n\tau$. Recently, it has been shown~\cite{sagraloff-complexity} that the bound on the bit complexity for the Descartes method can be lowered to $\tilde{O}(n^3\tau^2)$ when replacing exact computation by approximate computation (without abstaining from correctness). This result partially explains the success of such a modified Descartes method in practice.
However, as long as we restrict to the bisection strategy, it seems that the latter bound is optimal. We remark that Sch\"onhage already made a similar observation: In the introduction of~\cite{schonhage:fundamental}, he argued that "a factor $\tau^2$ inevitably occurs if nothing better than linear convergence is achieved".\\

In this paper, we introduce an exact and complete subdivision algorithm denoted $\dcn$ to isolate all real roots of a square-free polynomial with integer coefficients. 
Similar to the classical Descartes method, we use Descartes' Rule of Signs to determine an upper bound $v_I=\var(f,I)$ for the number of real roots of the polynomial $f$ within an interval $I$ that is actually processed. However, instead of splitting $I$ into two equally sized subintervals in each iteration, we consider a subdivision strategy that is based on Newton iteration: The analysis of the classical approach which exclusively uses bisection shows that the induced recursion tree is large ($\approx n\tau$) if and only if there exists a long sequences $I_1\supset I_2\supset\cdots\supset I_s$ of intervals in the subdivision process, where $v=v_{I_1}=\cdots=v_{I_s}$. Such a sequence implies the existence of a cluster $\mathcal{C}$ of $v$ nearby roots, and vice versa (cf. Theorem~\ref{Obreshkoff}). Hence, it seems reasonable to obtain a good approximation (i.e., an interval $I'\subset I$ close to $\mathcal{C}$) of such a cluster by considering a corresponding Newton step to approximate a $v$-fold root. Combining Descartes' Rule of Signs and a subdivision technique similar to the one as proposed by J.\,Abbott for the QIR method~\cite{abbott-quadratic}, we formulate a method to determine whether the so-obtained approximation $I'$ should be kept or not. In case of success, we proceed with the considerably smaller interval $I'$, whereas we fall back to bisection if the Newton step fails. 
Our analysis shows that, following this approach, we achieve quadratic convergence in most iterations. As a consequence, the induced recursion tree has almost optimal size $O(n\log(n\tau))$ which improves upon the bisection strategy by a factor of $\tau$. We further provide a detailed bit complexity analysis which yields the bound $\tilde{O}(n^3\tau)$ for $\dcn$. This matches the record bound achieved by the aforementioned asymptotically fast algorithms.

We consider our contribution of great importance because of the following reasons: (1) Although the proposed method is rather simple, it achieves the best bounds known for the bit complexity of the problem of isolating the real roots of a polynomial. (2) In addition, it is much easier to access and also much easier to implement than the asymptotically fast algorithms that are available so far. In comparison to the existing practical methods for real root isolation, the modifications are moderate, and thus we expect that a careful implementation of our new approach will outperform the existing ones. (3) Finally, our method can be applied to search for the real roots in some specified neighborhood of interest, a property which is not fulfilled by the algorithms as proposed by Pan and Sch\"onhage.

\section{Overall Idea}\label{sec:idea}

In this section, we first provide a high-level description of our new algorithm, and then outline the argument why this approach improves upon existing methods such as the classical Descartes method. For the exact definition of the algorithm and a detailed bit complexity analysis, we refer the interested reader to Section~\ref{sec:algorithm}.
Throughout the following considerations, let 
\begin{align}
f(x):=\sum_{i=0}^n a_i x^i\in\Z[x]\text{, with }|a_i|<2^\tau,\label{polyf}
\end{align}
be a square-free polynomial of degree $n$ with integer coefficients of bit-length $\tau$ or less. We further denote $z_1,\ldots,z_n$ the complex roots of $f$, $\sigma(z_i):=\min_{j\neq i}|z_i-z_j|$ the \emph{separation of $z_i$}, and $\sigma_f:=\min_i \sigma(z_i)$ the \emph{separation of $f$}. According to Cauchy's bound, the modulus of each root $z_i$ is bounded by $1+2^{\tau}\le 2^{\tau+1}$, and thus, for the task of isolating the real roots of $f$, we can restrict our search to the interval $\mathcal{I}_0:=(-2^{\tau+1},2^{\tau+1})$.

\begin{figure}[t]
\begin{center}
\includegraphics[width=11cm]{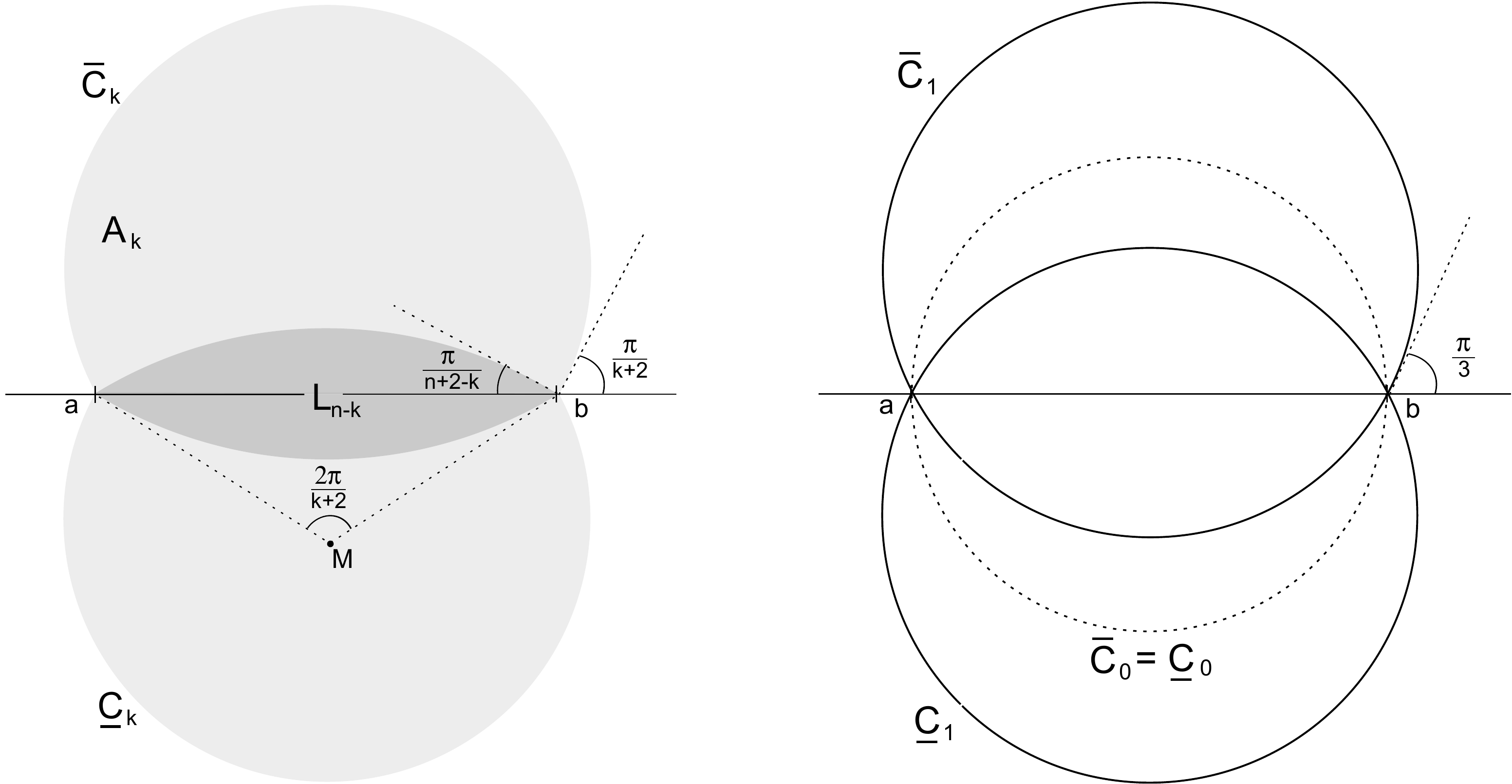}\end{center}
\caption{\label{fig:Obreshkoff} For any $k$ with $0 \le k \le n$, the \emph{Obreshkoff discs} 
$\overline{C}_k$ and $\underline{C}_k$ for $I=(a,b)$ have the endpoints of $I$ on their
boundaries; their centers see the line segment $(a,b)$ under the angle $2\pi/(k+2)$.  
The \emph{Obreshkoff lens} $L_k$ is the interior of $\overline{C}_k \cap
\underline{C}_k$, and the \emph{Obreshkoff area} $A_k$ is the interior of $\overline{C}_k \cup
\underline{C}_k$. Any point (except for $a$ and $b$) on the boundary of $A_k$ sees $I$ under an
angle $\pi/(k+2)$, and any point (except for
$a$ and $b$) on the boundary of
$L_k$ sees $I$ under the angle $\pi - \pi/(k+2)$. We have $L_n \subset \ldots \subset L_1 \subset L_0$ and $A_0
\subset A_1 \subset \ldots \subset A_n$. The cases $k=0$ and $k=1$ are of special interest: The circles
$\overline{C}_0$ and $\underline{C}_0$ coincide. They have their centers at the
midpoint of $I$. The circles $\overline{C}_1$
and $\underline{C}_1$ are the circumcircles of the two equilateral triangles
having $I$ as one of their edges. We call $A_0$ and $A_1$ the \emph{one} and \emph{two-circle regions} for $I$.}
\end{figure}

We consider an arbitrary root isolation method denoted $\textsc{Iso}$ which recursively performs subdivision on $I_0$ in order to determine isolating intervals for the real roots of $f$. \textsc{Iso} uses a counting function $\var(f,I)$ for $m$, the number of roots within an interval $I\subset \mathcal{I}_0$, where $v:=\var(f,I)\in\N$ fulfills the following properties:\vspace{0.25cm}
\begin{itemize}
\item[] (P1)\quad $v$ is an upper bound for $m$ (i.e., $v\ge m$), and
\item[] (P2)\quad $v$ has the same parity as $m$ (i.e., $v\equiv m \text{ }\operatorname{mod}\text{ }2$).\vspace{0.25cm}
\end{itemize}
The latter two properties imply that $\var(f,I)=m$ if $v\le 1$. Hence, in each step of the recursion, an interval $I$ is stored as isolating if $v=1$, and $I$ is discarded if $v=0$. If $v>1$, $I=(a,b)$ is subdivided (according to some subdivision strategy) into subintervals $I_1=(a,\lambda_1)$, $I_2=(\lambda_1,\lambda_2),\ldots,I_l:=(\lambda_{l-1},b)$, where $1\le l\le l_0$, $l_0\in\N$ is a global constant, and the $\lambda_i$ are rational values.
In order to detect roots at the subdivision points $\lambda_i$, we also check whether $f(\lambda_i)=0$ or not. $T_\textsc{Iso}$ denotes the recursion tree induced by $\textsc{Iso}$. Throughout the following considerations, we often treat nodes of $\textsc{Iso}$ and intervals produced by $\textsc{Iso}$ as interchangeable.

The definition of the counting function $\var(f,I)$ is based on Descartes' Rule of Signs: \textit{For an arbitrary polynomial $p=\sum_{i=0}^n p_i x^i\in\R[x]$, the number $m$ of positive real roots of $p$ is bounded by the number $v$ of sign variations in its coefficient sequence $(p_0,\ldots,p_n)$ and, in addition, $v\equiv m \text{ }\operatorname{mod}\text{ }2$}. In order to extend the latter rule to arbitrary intervals $I=(a,b)$, the M\"obius transformation $x\mapsto \frac{ax+b}{x+1}$ which maps $(0,+\infty)$ one-to-one onto $I$ is considered. Thus, for
\begin{align}
f_I(x)=\sum_{i=0}^n c_i x^i:=(x+1)^n\cdot f\left(\frac{ax+b}{x+1}\right),\label{poly:fI}
\end{align}
and $\var(f,I)$ defined as the number of sign variations in the coefficient sequence $(c_0,\ldots,c_n)$ of 
$f_I$, $\var(f,I)$ fulfills the properties (P1) and 
(P2). Because of the latter two properties and the fact that we never discard intervals that contain a real root of $f$, correctness of $\textsc{Iso}$ follows immediately.

The \emph{classical Descartes method} (\textsc{Dsc} for short) is a subdivision method which uses bisection in each iteration, that is, in each step, we have $l=l_0=2$, $I_1=(a,\lambda_1):=(a,m(I))$ and $I_2=(\lambda_1,b)=(m(I),b)$, with $m(I):=(a+b)/2$ the midpoint of $I$. Termination and complexity
analysis of \textsc{Dsc} rest on the following theorem:

\begin{thm}[\cite{Obrechkoff:book-english,Obreschkoff:book}]
\label{Obreshkoff} Let $I=(a,b)$ be an open interval and
$v= \var(f,I)$. If the Obreshkoff lens 
$L_{n - k}\subset\C$ (see Figure~\ref{fig:Obreshkoff} for the definition of $L_{n - k}$) contains at least $k$ roots (counted with
multiplicity) of $f$, then $v \ge k$. If the Obreshkoff area $A_k\subset\C$ contains at most $k$ roots (counted with multiplicity) of $k$, then $v \le k$. In
particular,\vspace{0.25cm}
\begin{itemize}
\item[](\textnormal{P3)}\quad\quad\quad\# of roots of $f$ in $L_n\le    \var(p,I) \le \#$ of roots
of $f$ in $A_n$.
\end{itemize}
\end{thm}

We remark that the special cases $k=0$ and $k=1$ appear as the one- and two-circle theorems in the literature~\cite{Alesina-Galuzzi,eigenwillig:thesis,krandick-mehlhorn:06,Obreshkoff25,Ostrowski:1950}. For the Descartes method, Theorem~\ref{Obreshkoff} implies that no interval $I$ of length $w(I)\le\sigma_f$ is split. Namely, its one-circle region $A_0$ cannot contain two or more roots. If $A_0$ contains no root, then $\var(f,I)=0$. Otherwise, $A_0$ contains one real root, and thus the two-circle region $A_1$ contains no non-real root. Hence, in the latter situation, we have $\var(f,I)=1$ by Theorem~\ref{Obreshkoff}. We conclude that the depth of the recursion tree $T_{\textsc{Dsc}}$ induced by the Descartes method is bounded by $\log w(\mathcal{I}_0)+\log\sigma_f^{-1}=\tau+\log\sigma_f^{-1}+2$. Furthermore, it holds (see~\cite[Corollary 2.27]{eigenwillig:thesis} for a self-contained proof):

\begin{thm}\label{subad}
Let $I$ be an interval and $I_1$ and $I_2$ be two disjoint subintervals of $I$. Then,\vspace{0.25cm}
\begin{itemize}
\item[]\textnormal{(P4)}\quad\quad\quad\quad\quad\quad\quad\quad\quad$\var(f,I_1) + var(f,I_2) \le \var(f,I).$
\end{itemize}
\end{thm}

\begin{figure}[t]
\begin{center}
\includegraphics[width=13cm]{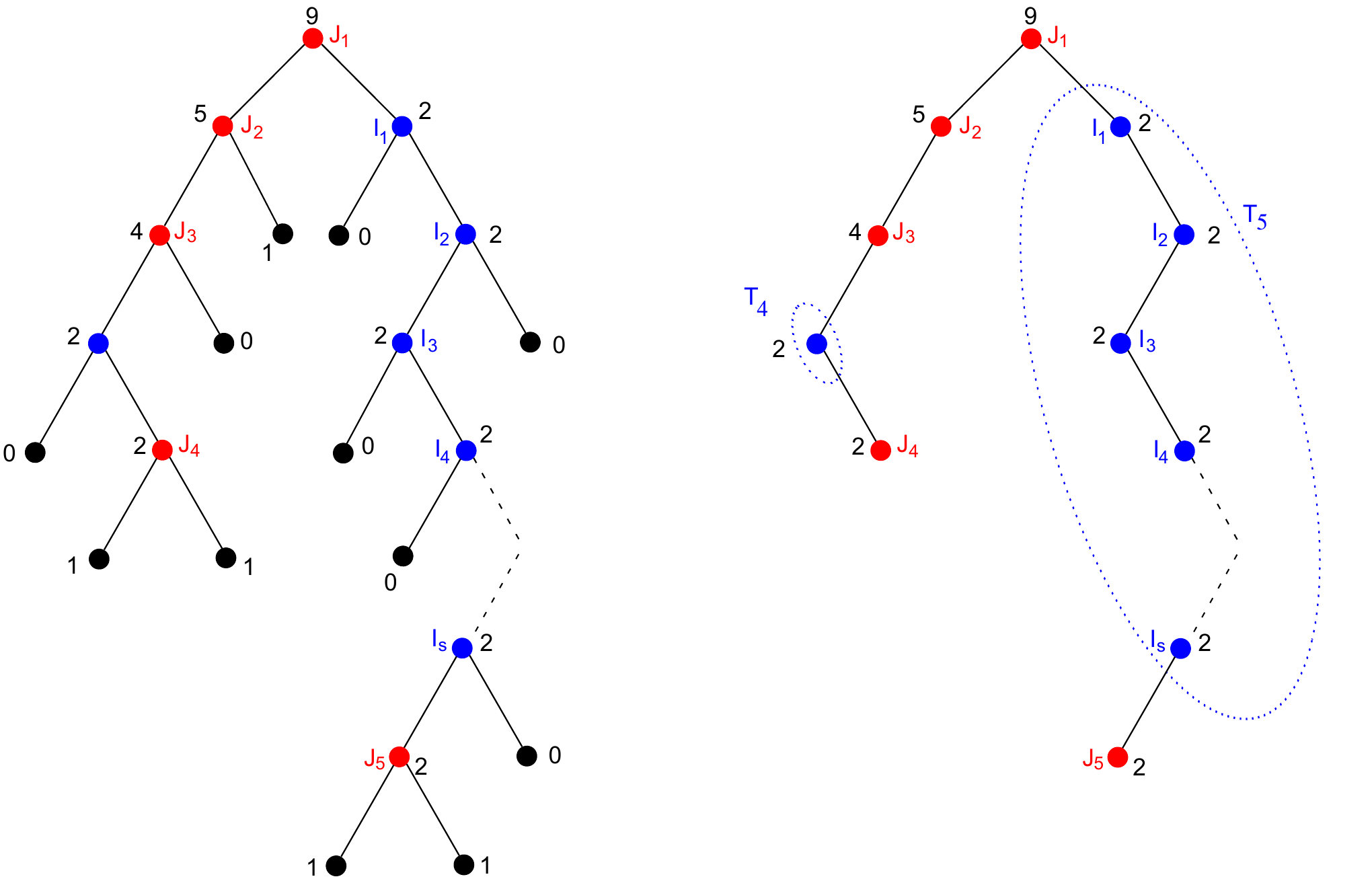}\end{center}
\caption{\label{tree} The left figure shows the subdivision tree $T_{\textsc{Dsc}}$ induced by the Descartes method, where, for each node $I$, the number $\var(f,I)$ of sign variations is given (e.g., $\var(f,J_1)=9$, or $\var(f,J_4)=2$). The colors \emph{red, black} and \emph{blue} indicate the (1) special nodes, (2) terminal nodes, and (3) all other nodes (non-special and non-terminal), respectively. The right figure shows the subtree $T_{\textsc{Dsc}}^*$ obtained by removing all terminal nodes. The non-special nodes in $T_{\textsc{Dsc}}^*$ partition into maximal connected components $T_4$ and $T_5$.\ignore{The long sequence $I_1\supset I_2\supset\cdots\supset I_s$ of intervals in $T_5$ with $\var(f,I_1)=\cdots=\var(f,I_s)=2$ corresponds to a large number of bisection steps to isolate two very nearby roots from each other; see the figure on the right with the graph of $f$ above $I_1$.}}
\end{figure}

According to Theorem~\ref{subad}, there cannot be more than $n/2$ intervals $I$ with $\var(f,I) \ge 2$ at any level of the recursion. Hence, the size of $T_{\textsc{Dsc}}$ is bounded by $n(\tau+\log\sigma_f^{-1}+2)$. Using Davenport-Mahler bound, one can further show~\cite{eigenwillig:thesis,sagraloff-complexity} that $\log\sigma_f^{-1}=O(n(\log n+\tau))$, and thus the bound for $|T_{\textsc{Dsc}}|$ writes as $\Otilde(n^2\tau)$. A more refined argument~\cite{eigenwillig:thesis} yields $|T_{\textsc{Dsc}}|=\Otilde(n\tau)$ which is optimal for the bisection strategy.\\

In the next step, we study the situation where the recursion tree $T_{\textsc{Dsc}}$ for the classical Descartes method is large. We then introduce our new algorithm which we denote $\textsc{Dsc}^2$ due to the quadratic convergence in most steps. $\textsc{Dsc}^2$ is a variant of the Descartes method which adaptively addresses the latter situation via combining Newton iteration and bisection. We also sketch the argument why this approach improves upon \textsc{Dsc}. The following definition is essential for the argument; see also Figure~\ref{tree}:

\begin{definition}
Let $I$ be a node (interval) in the recursion tree $T_{\textsc{Iso}}$ induced by some subdivision algorithm \textsc{Iso}. We call $I$ \emph{terminal} if $\var(f,I)\le 1$. A non-terminal interval $I$ with children $I_1,\ldots,I_l$ is called \emph{special} if $I=\mathcal{I}_0$ (i.e., $I$ is the root of $T_{\textsc{Iso}}$), or
\[
\var(f,I_j)<\var(f,I)\text{ for all }j=1,\ldots,s.
\]
\end{definition}

According to (P4) in Theorem~\ref{subad}, we have $\sum_{j=1}^l\var(f,I_j)\le\var(f,I)$ for each $I$. Thus, a non-terminal node $I$ different from $\mathcal{I}_0$ is non-special if and only if, for one of its children, we count the same number of sign variations as for $I$ and, for all other children, we count no sign variation. In total, there exist $n'$ special nodes, where $n'\le \var(f,\mathcal{I}_0)\le n$. Namely, when we subdivide a special interval which is not the root of the recursion tree, the non-negative value $\mu:=\sum_I\var(f,I)-\#\{I:\var(f,I)>0\}$ decreases by at least one, where we sum over all leafs in the actual iteration, and $\mu$ is initially set to $\mu=\var(f,\mathcal{I}_0)-1\le n-1$. We denote the special nodes by $J_1,\ldots,J_{n'}$ and assume, w.l.o.g., that $w(J_i)\ge w(J_k)$ if $i<k$. In particular, $J_1=\mathcal{I}_0$. We define $T_{\textsc{Iso}}^*$ the subtree of $T_{\textsc{Iso}}$ obtained from $T_{\textsc{Iso}}$ via removing all terminal nodes. Then, $T_{\textsc{Iso}}^*$ partitions into\vspace{0.25cm}
\begin{itemize}
\item[](1)\quad the special nodes $J_1,\ldots,J_{n'}$ (red dots in Figure~\ref{tree}), and
\item[](2)\quad subtrees $T_i\subset T_{\textsc{Iso}}^*$, with $i=2,\ldots,n'$, consisting of all non-special nodes $I\in  T_{\textsc{Iso}}^*$ with $J_i\subset I$ and $J_k \not\subset I$ for all special nodes $J_k$ with $J_k\supsetneq J_i$ (blue dots).\vspace{0.25cm}
\end{itemize}
From our definition of a special node, it follows that each $T_i$ constitutes a chain of intervals $I_1\supset\cdots\supset I_s$ that connects two special nodes. More precisely, $T_i$ connects $J_i$ with $J_k$, where $J_k$ is the special node of minimal width that contains $J_i$.  
Since each interval has at most $l_0$ children, $|T_{\textsc{Iso}}|$ is bounded by $(l_0+1)\cdot|T_{\textsc{Iso}}^*|$. Hence, we have
\begin{align}
O(|T_{\textsc{Iso}}|)=O(|T_{\textsc{Iso}}^*|)=O(n'+\sum_{i=2}^{n'} |T_i|)=O(n)+O(\sum_{i=2}^{n'} |T_i|).\label{treesize1}
\end{align}

\begin{figure}[t]
\begin{center}
\includegraphics[width=11cm]{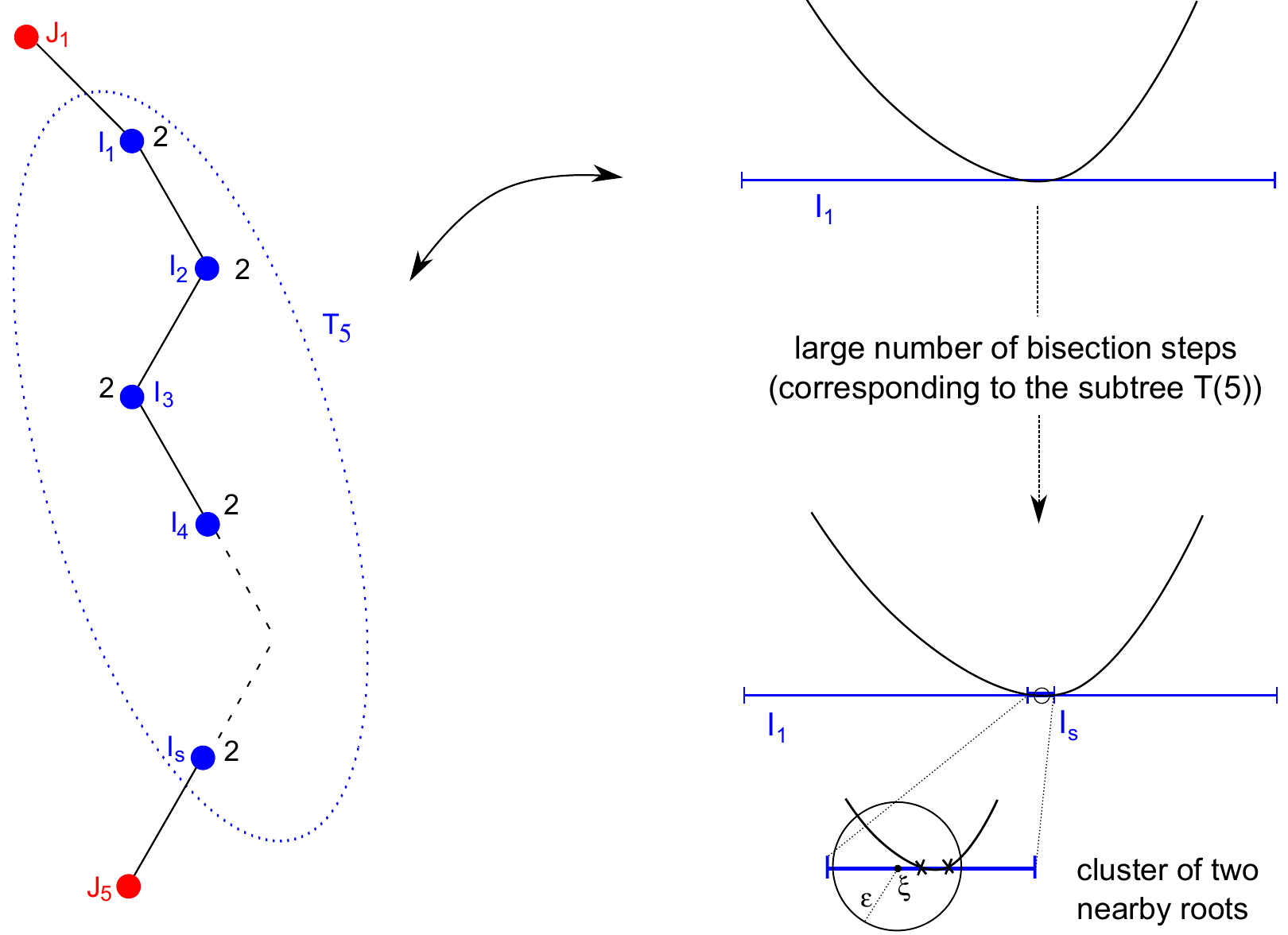}\end{center}
\caption{\label{fig:tree2} The long chain $I_1\supset I_2\supset\cdots\supset I_s$ of intervals in $T(5)$ with $\var(f,I_1)=\cdots=\var(f,I_s)=2$ corresponds to a large number of bisection steps to isolate two very nearby roots from each other; see the figure on the right with the graph of $f$ over $I_1$.}\end{figure}

The latter consideration shows that the size of the subdivision tree mainly depends on the length of the chains $T_i$. For the Descartes method, it might happen that some of these chains are very large (i.e., $|T_i|\approx n\tau$) which is due to the following situation (see also Figure~\ref{fig:tree2}): For a polynomial $f$ as in (\ref{polyf}), it is possible that there exists a $\xi\in\R$ and a very small, complex neighborhood (of size $\epsilon\approx 2^{-n\tau}$) of $\xi$ that contains a cluster $\mathcal{C}$ of $v$ nearby roots of $f$. Thus, separating these roots from each other via bisection requires at least $\log\epsilon^{-1}\approx n\tau$ steps. Furthermore, due to (P3) in Theorem~\ref{Obreshkoff}, there exists a long sequence $I_1\supset I_2\supset\cdots\supset I_s$ of non-special intervals with $\xi\in I_j$ for all $j$, and thus the number 
$$v:=\var(f,I_1)=\var(f,I_2)=\cdots=\var(f,I_s)$$ 
of sign variations does not change for the intervals in this sequence. Namely, for each $I_j$ in the above sequence, the Obreshkoff lens $L_n$ contains $\mathcal{C}$. Vice versa, according to Theorem~\ref{Obreshkoff}, such a long sequence of non-special intervals implies the existence of a cluster $\mathcal{C}$ consisting of $v$ nearby roots as above because the Obreshkoff area $A_n$ of each $I_j$ must contain at least $v$ roots.\footnote{The thoughtful reader may notice that the latter two statements are not completely rigorous: In particular, in the special case where $\xi$, and thus also the cluster $\mathcal{C}$, is very close to one of the endpoints of some $I_j$, it might happen that some of the roots are not considered by the counting function $\var(f,I_j)$ since they are located outside the Obreshkoff lens/area. We will address this issue in our algorithm as defined in Section~\ref{sec:algorithm}.} 
Since a cluster $\mathcal{C}$ of $v$ nearby roots at $\xi$ behaves very similar to a $v$-fold root at $\xi$, it seems reasonable to obtain a good approximation of $\mathcal{C}$ by considering Newton iteration instead of bisection. Namely, for a polynomial $p(x)\in\R[x]$ with a $v$-fold root at $\xi$ and a starting value $x_0$ sufficiently close to $\xi$, it is well-known from numerical analysis that the sequence $(x_i)_{i\in\N_0}$ recursively defined by
\begin{align*}
x_{i+1}:=x_i-v\cdot\frac{p(x_i)}{p'(x_i)}
\end{align*}
converges quadratically to $\xi$. 
Unfortunately, when isolating the roots of $f$, the situation differs considerably from the latter one: First, the above result only holds for a $v$-fold root $\xi$ and does not directly extend to a cluster $\mathcal{C}$ of $v$ roots near $\xi$. Second, in an early stage of the subdivision process, the existence of such a cluster $\mathcal{C}$ is not guaranteed, and even if one exists, we do not know what "sufficiently close to $\xi$" means in this situation.
 
In order to address the above mentioned problems and to finally turn the purely numerical Newton method into an exact and complete algorithm, we propose the following approach: Let $v=\var(f,I)$ be the number of sign variation for an actual interval $I=(a,b)$ in a certain iteration. Then, we consider this as an indicator that there might exist a cluster of $v$ nearby roots. Thus, we compute $\lambda:=t-v\cdot f(t)/f'(t)$ for some $t\in [a,b]$ (e.g., an endpoint of $I$) and consider an interval $I'=(a',b')\subset I$ of width $w(I')\ll w(I)$ that contains $\lambda$. If $\var(f,I')=v$ as well, we keep $I'$ and discard the intervals $(a,a']$ and $[b',b)$. Otherwise, we split $I$ into two equally sized intervals $I_1:=(a,m(I))$ and $I_2:=(m(I),b)$ and finally check whether $f(m(I))=0$ or not. Following this approach, no root is lost and intervals are at least bisected in each iteration. Furthermore, if a cluster $\mathcal{C}$ of nearby roots actually exists, we can hope to achieve fast convergence to this cluster when choosing $I'$ in an appropriate manner. In our algorithm, we choose $I'$ in a similar way as proposed by Abbott~\cite{abbott-quadratic} for the task of further refining intervals which are already isolating for an ordinary root. Namely, we decompose $I$ into a certain number $N_I$ of subintervals and pick the subinterval $I'$ of size $w(I)/N_I$ which contains $\lambda$. If $\var(f,I')=v$, then we keep $I'$ and decompose $I'$ into $N_{I'}=N_I^2$ subintervals in the next iteration. Otherwise, we continue with the intervals $I_1=(a,m(I))$ and $I_2=(m(I),b)$ which are now decomposed into only $N_{I_1}=N_{I_2}:=\max(4,\sqrt{N_I})$ many subintervals, etc.

In the next section, we give the exact definition of our new algorithm, and we show that it induces a subdivision tree of considerably smaller size than $T_{\textsc{Dsc}}$. In particular, it turns out that the size of each $T_i\subset T_{\textsc{Dsc}^2}$ is bounded by $O(\log n+\log \tau)$ which is due to the fact that, for most iterations, we have quadratic convergence to the real roots, and the width of each interval is lower bounded by $2^{-\tilde{O}(n\tau)}$; see Lemma~\ref{lemma:bounds1} and  Theorem~\ref{thm:main1} for proofs. Hence, according to (\ref{treesize1}), the size of the overall recursion tree is bounded by \begin{align}
O(n'\cdot(\log n\tau))=O(\var(f,\mathcal{I}_0)\cdot (\log n\tau))=O(n\cdot(\log n\tau)).\label{res:treesize}
\end{align}
The latter result particularly shows that the size of the recursion tree is directly correlated to the number $n^*$ of non-zero coefficients of $f$ because instead of considering $\mathcal{I}_0=(-2^{\tau+1},2^{\tau+1})$, we can start with $(-2^{\tau+1},0)$ and $(0,2^{\tau+1}))$, and the total number of sign variations counted for both intervals is upper bounded by $2\cdot n^*$. 

For the bit complexity of our algorithm, we have to consider the costs for computing the polynomials $f_I(x)=(1+x)^n\cdot f((ax+b)/(1+x))$ as defined in (\ref{poly:fI}), where $I=(a,b)$ is an interval to be processed. In Section~\ref{sec:bitcomplexity}, we will show that the costs for the latter step are dominated (up to constant factors) by the computation of $f(x+a)$. For $I\in T_i$, the endpoints of $I$ are dyadic numbers of bitsize $O(\tau+\log w(J_i)^{-1})$ or less, and thus the computation of $f_I$ demands for $\tilde{O}(n^2(\tau+\log w(J_i)^{-1}))$ bit operations. In Lemma~\ref{boundIj}, we prove that we can order the roots $z_1,\ldots,z_n$ in a way such that $\log w(J_i)^{-1}<\log\sigma(z_i)^{-1}+O(\log^2 n)$ for all $i=1,\ldots,n'$. Then, it follows that computing $f_I$ demands for at most $$\tilde{O}(n^2(\tau+\log w(J_i)^{-1}))=\tilde{O}(n^2(\tau+\log\sigma(z_i)^{-1}+\log^2 n))=\tilde{O}(n^2(\tau+\log\sigma(z_i)^{-1}))$$
bit operations.
Thus, for the total cost, we obtain the bound
\begin{align*} \tilde{O}(n^3\tau+n^2\sum_{i=1}^n\log\sigma(z_i)^{-1})=\tilde{O}(n^3\tau)
\end{align*}
since $\sum_{i=1}^n\log\sigma(z_i)^{-1}=O(n\tau(\log n\tau))=\tilde{O}(n\tau)$ according to Lemma 19 in~\cite{sagraloff-complexity}.
 
\section{Algorithm and Analysis}\label{sec:algorithm}

\subsection{The Algorithm}

We first present our new algorithm denoted $\dcn$. For pseudo-code, we refer to the Appendix.\vspace{0.15cm}
\hrule\vspace{0.1cm}
$\dcn$ maintains a list $\mathcal{A}$ of active intervals $I$ with corresponding integers $N_I=2^{2^{n_I}}$, $n_I\in\N_{\ge 1}$, and a list $\mathcal{O}$ of isolating intervals, where we initially set $\mathcal{A}:=\{(\mathcal{I}_0,4)\}:=\{((-2^{\tau+1},2^{\tau+1}),4)\}$ and $\mathcal{O}:=\emptyset$. For $(I,N_I)\in\mathcal{A}$, $I=(a,b)$, we proceed as follows: We remove $I$ from $\mathcal{A}$ and compute the number $v:=\var(f,I)$ of sign variations for $f$ on $I$. 

\begin{enumerate} 
\item If $v=0$, we do nothing (i.e., $I$ is discarded). 
\item If $v=1$, then $I$ isolates a real root of $f$. Thus, we add $I$ to the list $\mathcal{O}$ of isolating intervals. 
\item For $v>1$, we proceed as follows:\vspace{0.25cm}
\begin{enumerate} 
\item For $i=1,2$, let 
\begin{align}
B_1:=(a,a+\frac{w(I)}{N_I})\quad\text{ and }\quad B_2:=(b-\frac{w(I)}{N_I},b)\label{def:B12}
\end{align}
be the left- and rightmost interval of size $w(I)/N_I$ contained in $I$, respectively. We compute $v_i:=\var(f,B_i)$: If one of the values $v_1$ or $v_2$ equals $v$, then the corresponding interval $B_i$ (at most one of the two disjoint intervals $B_i$ fulfills $\var(f,B_i)=v$) contains all roots of $f$ within $I$. Hence, we keep $I':=B_i$, discard $I\backslash I'$, and set $N_{I'}:=N_I^2$. That is, $(I',N_{I'}):=(B_i,N_I^2)$ is added to $\mathcal{A}$.
\item If both values $v_1$ and $v_2$ differ from $v$, we compute
\begin{align}
\lambda_1:=a-v\cdot\frac{f(a)}{f'(a)}\quad\text{and}\quad\lambda_2:=b-v\cdot\frac{f(b)}{f'(b)}.\label{def:lambda}
\end{align}
If $I$ contains a cluster $\mathcal{C}$ of $v$ nearby roots and $a$ (or $b$) has "reasonable" distance to $\mathcal{C}$, then $\lambda_1$ (or $\lambda_2$) constitutes a considerably better approximation of $\mathcal{C}$ than $a$ (or $b$). We check whether this is actually the case: For $i=1,2$, we first compute the point $a+k_i\cdot \frac{w(I)}{4N_I}$, with $k_i\in\{2,\ldots,4N_I-2\}$, which is closest to $\lambda_i$ (if there exist two equally close points, we choose the one with smaller index). In more mathematical terms,
\begin{align}
k_i:=\min(\max(\lfloor 4N_I(\lambda_i-a)\rfloor,2),4N_I-2).
\end{align}
Then, we define $I'_i$ to be the interval of length $w(I)/N_I$ centered at the subdivision point $a+k_i\cdot w(I)/4N_I$, that is, 
\begin{align}
I'_i:=(a+(k_i-2)\cdot \frac{w(I)}{4N_I},a+(k_i+2)\cdot \frac{w(I)}{4N_I})\subset I.\label{def:Ii}
\end{align} 
In particular, if $a+w(I)/4N_I\le \lambda_i\le b-w(I)/4N_I$, then $I'_i$ contains $\lambda_i$, and $\lambda_i$ has distance at least $w(I)/4N_I$ to both endpoints of $I_i'$. Now, we compute $v_1':=\var(f,I'_1)$ and $v_2':=\var(f,I'_2)$. If one of the two values $v_1'$ or $v_2'$ equals $v$, we keep the corresponding interval $I':=I'_i$ with $v_i'=v$ (if we count $v$ sign variations for both intervals $I_1'$ and $I_2'$, we just keep $I_1'$) and discard $I\backslash I'$. Finally, we add $(I',N_{I'}):=(I',N_I^2)$ to $\mathcal{A}$. 
\item If all values $v_1$, $v_2$, $v_1'$ and $v_2'$ differ from $v$, then we consider this an indicator that there is either no cluster of $v$ nearby roots or that such a cluster is not separated well enough from the remaining roots. Hence, in this situation, we fall back to bisection. That is, we split $I$ into two equally sized intervals $I_1=(a,m(I))$ and $I_2=(m(I),b)$ and add $(I_1,\max(4,\sqrt{N}))$ and $(I_2,\max(4,\sqrt{N}))$ to $\mathcal{A}$. Finally, if $f(m(I))=0$, we also add $[m(I),m(I)]$ to $\mathcal{O}$.
\end{enumerate}
\end{enumerate}\vspace{0.1cm}
\hrule\newpage

Correctness and termination of $\dcn$ are obvious because our starting interval $\mathcal{I}_0$ contains all real roots of $f$, we never discard intervals (or endpoints) that contain a root of $f$, and intervals are at least bisected in each iteration. In addition, we obtain the following bounds on the width of the intervals $I$ and the corresponding numbers $N_I$ produced by $\dcn$:

\begin{lemma}\label{lemma:bounds1}
For each interval $I$ produced by $\dcn$, we have $$2^{\tau+2}\ge w(I)\ge\sigma_f^3\cdot 2^{-2(\tau+2)}=2^{-\tilde{O}(n\tau)}\text{ and }4\le N_I\le 2^{2(\tau+2)}\cdot\sigma_f^{-2}=2^{\tilde{O}(n\tau)}.$$
In particular, for $I\in T_i$ (see Section~\ref{sec:idea} for the definition of the subtree $T_i\subset T_{\dcn}$), we have
$$2^{\tau+2}\ge w(I)\ge w(J_i)\text{ and }4\le N_I\le 2^{2(\tau+2)}\cdot w(J_i)^{-2},$$
with $J_i$ the special node corresponding to $T_i$.
\end{lemma}

\begin{proof}
The inequalities $2^{\tau+2}\ge w(I)$ and $N_I\ge 4$ are trivial. For $N_I>4$, there must exist an interval $J\supset I$ with $N_J=\sqrt{N_I}$, and $J$ was replaced by an interval $J'\supseteq I$ of size $w(J)/N_J$. Since $J$ is non-terminal, $J'$ is also non-terminal because $\var(f,J')=\var(f,J)>1$. Thus, $\sigma_f\le w(J')=w(J)/N_J\le 2^{\tau+2}/\sqrt{N_I}$. This shows the upper bound for $N_I$. For the lower bound for $w(I)$, we consider the parent interval $J$ of $I$. Since $J$ is non-terminal, we have $w(J)\ge \sigma_f$, and thus $w(I)\ge w(J)/N_J\ge \sigma_f\cdot \sigma_f^2\cdot 2^{-2(\tau+2)}$. For $I\in T_i$, the bounds for $w(I)$ are trivial, and, in completely similar manner as above, we conclude that $2^{\tau+2}/\sqrt{N_I}\ge w(J_i)$.
\end{proof}

Throughout the following considerations, we call a subdivision step from $I$ to $I'\subset I$ \emph{quadratic} if $w(I')=w(I)/N_I$, and we call a subdivision step \emph{linear} if $I$ is split into two equally sized intervals $I_1$ and $I_2$. In a quadratic step, the integer $N_I$ is squared whereas, in a linear step, $N_{I'}:=\max(4,\sqrt{N_I})$ for each of the subintervals $I'=I_{1/2}$.

\subsection{Analysis of the Recursion Tree}

In this section, we prove that the size of each of the subtrees $T_i\subset T_{\dcn}$ as defined in Section~\ref{sec:idea} is bounded by $O(\log n+\log \tau)$. We first have to investigate into the following technical lemmata:

\begin{lemma}\label{lemma:sequence}
Let $w$, $w'\in\R^+$ be two positive reals with $w>w'$, and let $m\in\mathbb{N}_{\ge 1}$ be a positive integer. The sequence $(s_i)_{i\in\N_{\ge 1}}:=((x_i,n_i))_{i\in\N_{\ge 1}}$ is recursively defined as follows: $s_1=(x_1,n_1)=(w,m)$, and
$$s_i=\left(x_i,n_i\right):=\begin{cases}
  \left(\frac{x_{i-1}}{N_{i-1}},n_{i-1}+1\right),  & \text{if }\frac{x_{i-1}}{N_{i-1}}\ge w'\\
  \left(\frac{x_{i-1}}{2},\max(1,n_{i-1}-1)\right), & \text{if }\frac{x_{i-1}}{N_{i-1}}<w',
\end{cases}
$$
where $N_i:=2^{2^{n_{i}}}$ and $i\ge 2$. Then, the smallest index $i_0$ with $x_{i_0}\le w'$ is upper bounded by $8(n_1+\log\log \max(4,\frac{w}{w'}))$.
\end{lemma}

\begin{proof}
Throughout the following consideration, we call an index $i$ \emph{strong} (\textbf{S}) if $x_i/N_i\ge w'$ and \emph{weak} (\textbf{W}), otherwise. If $w/4<w'$, then each $i\ge 1$ is weak, and thus $i_0\le 3$. For $w/4\ge w'$, let $k\in\mathbb{N}_{\ge 1}$ be the unique integer with $$2^{-2^{k+1}}<w'/w\le 2^{-2^{k}}.$$
Then, $k\le\log\log\frac{w}{w'}$, and since $x_i\le x_{i-1}/2$ for all $i$, there must exist an index $i$ which is weak. Let $k'$ denote the smallest weak index.\\

\noindent\emph{Claim 1:} $k'\le k+1$

\noindent Assume otherwise, then the indices $1$ to $k$ are all strong. Hence,
$$x_{k+1}=w\cdot 2^{-(2^{m}+2^{m+1}+\cdots+2^{m+k-1})}=
w\cdot 2^{-2^{m}(2^0+2^1+\cdots+2^{k-1})}=w\cdot 2^{-2^{m}(2^{k}-1)}\le 4w\cdot 2^{-2^{k+1}}<4w',
$$
and $n_{k+1}>1$. It follows that $k+1$ is weak, a contradiction.\\

Let us now consider the subsequence $\mathcal{S}=k',k'+1,\ldots,i_0-3$:\\

\noindent\emph{Claim 2:} $\mathcal{S}$ contains no subsequence of type ...\textbf{SS}... or ...\textbf{SWSWS}...

\noindent If there exists a weak index $i$ and two strong indices $i+1$ and $i+2$, then $x_{i}/N_i>x_{i+2}/N_{i}\ge x_{i+2}/N_{i+2}\ge w'$ contradicting the fact that $x_i/N_i<w'$. Since $\mathcal{S}$ starts with a weak index, the first part of our claim follows. For the second part, assume that $i$, $i+2$ and $i+4$ are strong. Then, $i+1$ and $i+3$ are weak, and thus
$$w'\le x_{i+4}/N_{i+4}<x_{i+2}/(N_{i+2}\cdot N_{i+4})<x_i/(N_i\cdot N_{i+2}\cdot N_{i+4})=x_i/N_i^3=x_{i+1}/N_{i+1}$$ contradicting the fact that $i+1$ is weak.\\ 

\noindent\emph{Claim 3:} If $i$ is weak and $i<i_0$, then $n_i\ge 2$.

\noindent Namely, if $i$ is weak and $n_i=1$, then $x_i/4=x_i/N_i<w'$, and thus $x_{i_0-1}<w'$ which contradicts the definition of $i_0$.\\

We now partition the sequence $\mathcal{S}$ into maximal subsequences $\mathcal{S}_1,\mathcal{S}_2,\ldots,\mathcal{S}_r$ such that each $\mathcal{S}_j$, $j=1,\ldots,r$, contains no two consecutive weak elements. Then, according to our above results, each $\mathcal{S}_j$, with $j<r$, is of type \textbf{W}, \textbf{WSW}, or \textbf{WSWSW}. The last subsequence $\mathcal{S}_r$ (with last index $i_0-3$) is of type \textbf{W}, \textbf{WS}, \textbf{WSW}, \textbf{WSWS}, or \textbf{WSWSW}. After each $S_j$, with $j<r$, the number $n_i$ decreases by one, and thus we must have $r\le n_1+k'$ since we start with $n_{k'}=n_1+k'-1$ and, in addition, $n_i\ge 2$ for all weak $i$. Since the length of each $\mathcal{S}_j$ is bounded by $5$, it follows that 
$$i_0= i_0-3 +3\le k'+5r+3\le 5(n_1+k')+k'+3\le 8(n_1+k).$$
\end{proof}

\begin{lemma}\label{lem:lensarea}
Let $I=(a,b)$ an arbitrary interval, $A_n$ the corresponding Obreshkoff area and $L_n$ the Obreshkoff lens for $I$. Then:\\\\
\noindent (1) For $I'=(a',b')\subset I$ with $a\neq a'$ and $b\neq b'$, the Obreshkoff area $A_n'$ for $I'$ is completely contained within the lens $L_n$ if $$\min(|a-a'|,|b-b'|)>8n^2 w(I').$$
In the latter situation, it holds that $$|x-\xi|>\frac{1}{4n}\cdot\left(\min(|a-a'|,|b-b'|)-8nw(I')\right)$$ for all $x\notin L_n$ and all $\xi\in A_n'$.\\\\   
\noindent (2) For $I'=(a',b')$ with $I'\cap I=\emptyset$, the Obreshkoff area $A_n'$ for $I'$ does not intersect $A_n$ if $$\operatorname{dist}(I,I')>4n^2\cdot\min(w(I),w(I')),$$ where $\operatorname{dist}(I,I')$ denotes the distance between the two intervals $I$ and $I'$.
\end{lemma}

\begin{proof}
(1) In a first step, we compute the radius $r'$ of the Obreshkoff discs $\underline{C}_n'$ and $\overline{C}_n'$ for the interval $I'=(a',b')$: A point $\xi$ on the boundary of the Obreshkoff area $A_n'$ (except $a'$ and $b'$) sees $I'$ under an angle $\gamma=\pi/(n+2)$; see Figure \ref{fig:Obreshkoff} and \ref{fig:Obreshkoff2}. Hence, from the extended Sine Theorem, it follows that $$r'=\frac{w(I')}{2\sin(\gamma)}=\frac{w(I')}{2\sin(\pi/(n+2)}<\frac{(n+2)w(I')}{\pi}<n\cdot w(I')$$
since $\sin x>x/2$ for all $x\in (0,\pi/4]$. In particular, each point $z$ within the Obreshkoff area $A_n'$ has distance at most $2r'<2n\cdot w(I')$ from any point within $I$. W.l.o.g, we assume that $|a-a'|\le |b-b'|$. Then, the distance from $a'$ to the boundary of the Obreshkoff lens $L_n$ for $I$ is bounded by the distance $\delta$ from $a'$ to the line $\overline{ac}$, where $c$ denotes the topmost point of $L_n$. Since $\overline{ac}$ intersects the $x$-axes in an angle of $\pi/(2(n+2))$, we have $\delta=|a-a'|\sin \pi/(2(n+2))>|a-a'|/(4n)$. Thus, $A_n'\subset L_n$ if $|a-a'|/(4n)>2n\cdot w(I')$ or $|a-a'|>8n^2w(I')$. For $|a-a'|>|b-b'|$, a similar argument shows that $A_n'\subset L_n$ if $|b-b'|>8n^2w(I')$. In addition, if $\min(|a-a'|,|b-b'|)>8n^2 w(I')$, then each point $\xi\in A_n'$ has distance at least $\delta-2nw(I')>\min(|a-a'|,|b-b'|)/(4n)-2nw(I')$ from any point located outside of $L_n$.\\ 
(2) W.l.o.g., we can assume that $w(I')\le w(I)$ and $a'\ge b$. Let $L$ be the line passing through $b$ which intersects the $x$-axes in an angle of $\pi/(n+2)$. Then, the upper part of the Obreshkoff area $A_n$ lies completely on one side of this line. Now, if $A_n'$ lies completely on the other side of $L$, then, by symmetry, $A_n$ and $A_n'$ do not share a common point. We have already argued that $A_n'$ is contained within the disc of radius $2nw(I')$ centered at $a'$. Hence, if the distance $\delta':=\operatorname{dist}(a',L)$ from $a'$ to $L$ is larger than $2nw(I')$, then $A_n\cap A_n'=\emptyset$. We have $\delta'=|a'-b|\sin(\pi/(n+2))>|a'-b|/2n=\operatorname{dist}(I,I')/2n$, and thus our claim follows.  
\end{proof}

\begin{figure}[t]
\begin{center}
\includegraphics[width=13.5cm]{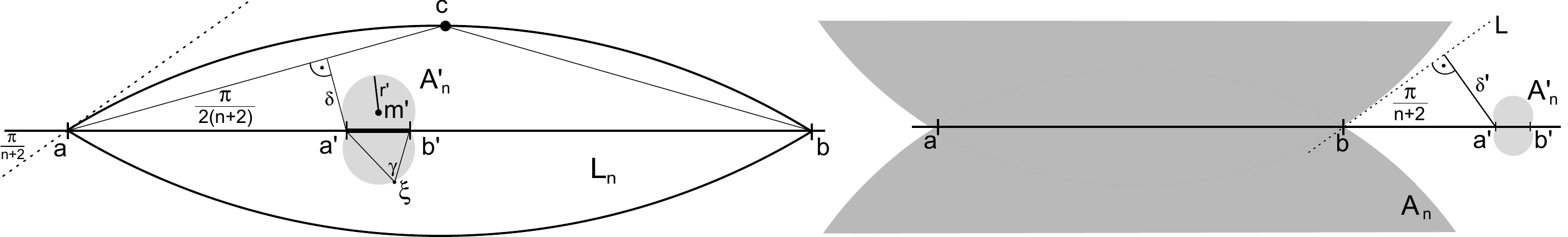}\end{center}
\caption{\label{fig:Obreshkoff2}On the left figure, $c$ denotes the topmost point of the Obreshkoff lens $L_n$ for $I=(a,b)$. If $|a-a'|\le w(I)/2$, then the distance from $a'$ to the boundary of $L_n$ is bounded by the distance $\delta$ from $a'$ to $\overline{ac}$. The radius $r'$ of the Obreshkoff discs $\underline{C}_n'$ and $\overline{C}_n'$ for $I'=(a',b')$ is bounded by $n\cdot w(I')$ due to the extended Sine Theorem. The right figure shows the Obreshkoff areas for the intervals $I$ and $I'$, respectively.}
\end{figure}

We now turn to the analysis of the subtrees $T_i\subset T_{\dcn}$ as defined in Section~\ref{sec:idea}. There, we have already argued that each $T_i$ constitutes a chain of intervals $I_{1}=(a_1,b_1)\supset I_2=(a_2,b_2)\supset\cdots\supset I_s=(a_s,b_s)$ ``connecting'' the special node $J_i$ with the special node $J_k$ of minimal length that contains $J_i$. In the following Theorem, we will show that, for all but $O(\log n)$ many $j$, the sequence $(w(I_j),n_{I_j})=(w(I_j),\log\log N_{I_j})$ behaves similar to the sequence $(x_j,n_j)$ as defined in Lemma~\ref{lemma:sequence}. As a result, we obtain the following bound on $|T_i|$:

\begin{theorem}\label{thm:main1}
For each special node $J_i$, the corresponding subtree $T_i\subset T_{\dcn}$ has size
$$
|T_i|=O(\log n+\log \tau).
$$  
\end{theorem}
\begin{proof}
We first consider the special case where $a_1=a_2=\cdots=a_s$, that is, in each subdivision step, the leftmost interval has been chosen. Since $v=\var(I_1)=\cdots=\var(I_s)$, Theorem~\ref{subad} implies that $\var(f,I)=v$ for each interval $I$ with $I_s\subset I\subset I_1$. In particular, if $w(I_j)/N_{I_j}\ge w(I_s)$, we count $v$ sign variations for the interval $B_1=(a_j,a_j+w(I_j)/N_{I_j})=(a_s,a_s+w(I_j)/N_{I_j})$ as defined in (\ref{def:B12}). Thus, the subdivision step from $I_j$ to $I_{j+1}$ is quadratic in this case. Then, for $j=1,\ldots,s-1$, the sequence $(w(I_j),n_{I_j})$ coincides with the sequence $(x_j,n_j)$ as defined in Lemma~\ref{lemma:sequence}, where $w:=w(I_1)$, $w':=w(I_s)$ and $n_1=m:=n_{I_1}$. Namely, if $w(I_j)/N_{I_j}\ge w'$, we have $w(I_{j+1})=w(I_j)/N_{I_j}$ and $n_{I_{j+1}}=1+n_{I_j}$, and, otherwise, we have $w(I_{j+1})=w(I_j)/2$ and $n_{I_{j+1}}=\max(1,n_{I_j}-1)$. Hence, according to Lemma~\ref{lemma:sequence}, it follows that $s$ is bounded by $$8(n_{I_1}+\log\log\max(4,w(I_1)/w(I_s)))=O(\log n+\log \tau),$$
where we used the bounds on $n_{I_1}$ and $w(I_s)$ from Lemma~\ref{lemma:bounds1}. An analogous argument shows the same bound for $s$ in the case where $b_1=b_2=\cdot b_s$.

We now turn to the more general case, where $a_1\neq a_s$ and $b_1\neq b_s$: Let $s_1\in\{1,\ldots,s\}$ be the smallest index with $a_{s_1}\neq a_1$ and $b_{s_1}\neq b_1$. Then, due to the above argument, $s_1$ is bounded by $O(\log n+\log\tau)$. Furthermore, $\min(|a_{1}-a_{s_1}|,|b_1-b_{s_1}|)\ge w(I_{s_1})/4$, and thus
\begin{align*}
&\min(|a_1-a_{j}|,|b_1-b_{j}|)\ge 2^{j-s_1-2}w(I_j)\\
&\hspace{-0.5cm}\Rightarrow \frac{1}{4n}\cdot\left(\min(|a_1-a_j|,|b_1-b_j|)-8nw(I_j)\right)\ge w(I_j)\left(\frac{2^{j-s_1-4}}{n}-8n\right)\text{ for all }j\ge s_1.
\end{align*}
Hence, with $s_2:=s_1+\lceil\log(16n^3)\rceil+4=O(s_1)+O(\log n)$, this yields
\begin{align*}
\frac{1}{4n}\cdot\left(\min(|a_1-a_{j}|,|b_1-b_{j}|)-8nw(I_{j})\right)\ge 8n^2 w(I_{j})\text{ for all }j\ge s_2.
\end{align*}
Then, from Theorem~\ref{Obreshkoff} and Lemma~\ref{lem:lensarea}, we conclude 
that, for $j\ge s_2$, the Obreshkoff area $A_{n}^{(j)}$ for $I_{j}$ contains 
exactly $v$ roots $z_1,\ldots,z_v$ of $f$ because the Obreshkoff lens $L_{n}^{(1)}$ for $I_{1}$ contains at most $v$ roots and $A^{(j)}_n\subset L^{(1)}_n$ contains at least $v$ roots. In particular, the Obreshkoff area $A^{(s)}_n$ for $I_s$ must contain $z_1,\ldots,z_v$. In the proof of Lemma~\ref{lem:lensarea}, we have already argued that each point within $A^{(s)}_n$ has distance less than $2nw(I_s)$ from any point within $I_s$, and thus
\begin{align}
|x-z_i|<2nw(I_s),\text{ for all }i=1,\ldots,v\text{ and all }x\in I_s.
\end{align}
The remaining roots 
$z_{v+1},\ldots,z_n$ of $f$ are located outside the Obreshkoff lens $L^{(1)}_n$ for 
$I_1$, and thus their distance to an arbitrary point within $I_j$ is larger than $8n^2w(I_{s_2})$. Namely, according to Lemma~\ref{lem:lensarea}, the distance from any of the roots $z_{v+1},\ldots,z_n$ to an arbitrary point within $I_{s_2}$ is lower bounded by $8n^2w(I_{s_2})$ and $I_{s_2}$ contains $I_j$.
The following consideration further shows the existence of an $s_3=s_2+O(\log n+\log\tau)$ such that $w(I_j)\le w(I_{s_2})/N_{I_j}$ for all $j\ge s_3$, and thus 
\begin{align}
|x-z_i|>8n^2 N_{I_j} w(I_j)\text{ for all }i=v+1,\ldots,n,\text{ }j\ge s_3,\text{ and all }x\in I_j:\label{outerbound}
\end{align}
Due to Lemma~\ref{lemma:bounds1}, we have $N_{I_j}\le N_{\max}:=\lceil 2^{2(\tau+2)}/\sigma_f\rceil=2^{\tilde{O}(n\tau)}$ for all $j$. Thus, if the sequence $I_{s_2},$ $I_{s_2+1},\ldots$ starts with more than $m_{\max}:=\log\log N_{\max}+1=O(\log n+\log\tau)$ consecutive linear subdivision steps, then $N_{I_{j'}}=4$ and $w(I_{j'})\le w(I_{s_2})/4=w(I_{s_1})/N_{I_{j'}}$ for some $j'\le s_2+\log\log m_{\max}$. Otherwise, there exists a $j'$ with $s_2\le j'\le s_2+m_{\max}$ such that the step from $I_{j'}$ to $I_{j'+1}$ is quadratic. Since the length of a sequence of consecutive quadratic subdivision steps is also bounded by $m_{\max}$, there must exist a $j''$ with $j'+1\le j''\le j'+m_{\max}+1$ such that the step from $I_{j''-1}$ to $I_{j''}$ is quadratic, whereas the step from $I_{j''}$ to $I_{j''+1}$ is linear. Then, $N_{I_{j''+1}}=\sqrt{N_{I_{j''}}}=N_{I_{j''-1}}$ and $$w(I_{j''+1})=w(I_{j''})/2=w(I_{j''-1})/(2N_{I_{j''-1}})<w(I_{s_2})/N_{I_{j''+1}}.$$ Hence, in both cases, we have shown that there exists an $s_3\le s_2+2m_{\max}+1=O(\log n+\log\tau)$ with $w(I_{s_3})\le w(I_{s_2})/N_{I_{s_3}}$. Then, by induction, it follows that $w(I_{j})\le w(I_{s_2})/N_{I_j}$ for all $j\ge s_3$ which shows (\ref{outerbound}).

We are now ready to show that the subdivision step from $I_j$ to $I_{j+1}$ is quadratic if $j\ge s_3$ and $w(I_j)\ge 68nN_{I_j} w(I_s)$: Namely, if the latter two inequalities hold, then one of the endpoints (w.l.o.g., we assume this point to be $a_j$) of $I_j$ has distance at least $w(I_j)/2\ge 34nN_{I_j}w(I_s)$ from $a_s$. Thus, the distance from $a_j$ to any of the roots $z_1,\ldots,z_v$ is larger than $34nN_{I_j}w(I_s)-2nw(I_s)\ge 32nN_{I_j}w(I_s)$. In addition, we have $|a_j-z_i|>8n^2 N_{I_j} w(I_{j})$ for all $i>v+1$ due to (\ref{outerbound}). Thus,\begin{align} \nonumber
\left|\frac{1}{v}\cdot\frac{(a_j-a_s)f'(a_j)}{f(a_j)}\right|&=\left|\frac{1}{v}\sum_{i=1}^v\frac{a_j-a_s}{a_j-z_i}+\frac{1}{v}\sum_{i=v+1}^n \frac{a_j-a_s}{a_j-z_i}\right|=\left|1+\frac{1}{v}\sum_{i=1}^v\frac{z_i-a_s}{a_j-z_i}+\frac{1}{v}\sum_{i=v+1}^n \frac{a_j-a_s}{a_j-z_i}\right|\\ \nonumber 
&\le 1+\frac{1}{v}\sum_{i=1}^v\frac{|z_i-a_s|}{|a_j-z_i|}+\frac{1}{v}\sum_{i=v+1}^n \frac{|a_j-a_s|}{|a_j-z_i|}<1+\frac{2nw(I_s)}{32nN_{I_j}w(I_s)}+n\cdot\frac{w(I_j)}{8n^2N_{I_j}w(I_j)}\\
\nonumber&\le 1+\frac{1}{16N_{I_j}}+\frac{1}{8nN_{I_j}}\le 1+\frac{1}{8N_{I_j}},
\end{align}
where we used that $f'(a)/f(a)=\sum_{i=1}^n (a-z_i)^{-1}$ for all $a\in\C$ with $f(a)\neq 0$.
In completely analogous manner, we show that
\begin{align*} 
\left|\frac{1}{v}\cdot\frac{(a_j-a_s)f'(a_j)}{f(a_j)}\right|>1-\frac{1}{8N_j}.
\end{align*}
This yields the existence of an $\epsilon\in\R$ with $|\epsilon|<1/(8N_j)\le 1/32$ and $\frac{1}{v}\cdot\frac{(a_j-a_s)f'(a_j)}{f(a_j)}=1+\epsilon$. We can now derive the following bound on the distance between the approximation obtained by the Newton iteration and $a_s$:
\begin{align}\nonumber
\left|a_s-(a_j-v\cdot\frac{f(a_j)}{f'(a_j)})\right|&=|a_s-a_j|\cdot \left|1-\frac{1}{\frac{1}{v}\cdot\frac{(a_j-a_s)f'(a_j)}{f(a_j)}}\right|=|a_s-a_j|\cdot \left|1-\frac{1}{1+\epsilon}\right|\\
&=\epsilon\cdot|a_s-a_j|\cdot\left|\frac{1}{1+\epsilon}\right|<\frac{4|a_s-a_j|}{33N_{I_j}}\le\frac{w(I_j)}{7N_{I_j}}\label{shrink}
\end{align}
If $a_s\ge b_j-w(I_j)/N_{I_j}$, then $(a_s,b_s)\subset B_2=(b_j-w(I_j)/N_{I_j},b_j)$ (cf. Step 3 (a), (\ref{def:B12}) in our algorithm for the definition of $B_2$), and thus $\var(f,B_2)=v$. Hence, in this case, we keep $I_{j+1}=B_2$ which has size $w(I_j)/N_{I_j}$. If $a_s<b_j-w(I_j)/N_{I_j}$, then according to (\ref{shrink}) we have $$a_j-v\cdot\frac{f(a_j)}{f'(a_j)}\in (a_s-\frac{w(I_j)}{7N_{I_j}},a_s+\frac{w(I_j)}{7N_{I_j}})\subset (a_j+\frac{w(I_j)}{4N_{I_j}},b_j-\frac{w(I_j)}{4N_{I_j}}).$$ 
It follows that the interval $I_1'$ as defined in Step 3 (b), (\ref{def:Ii}) of our algorithm contains $(a_s,b_s)$, and thus $\var(f,I_1')=v$. This shows that the subdivision step from $I_j$ to $I_{j+1}$ is quadratic.

We now consider the sequence $(w(I_{s_3+i}),n_{I_{s_3+i}})$, for $i=1,\ldots,i'$, where $i'$ is defined as the largest index with $w(I_{s_3+i'})\ge 68nw(I_s)$. Then, our above argument implies that the sequence $(w(I_{s_3+i}),n_{I_{s_3+i}})_{1\le i\le i'-1}$ coincides with the sequence $(x_i,n_i)_{1\le i\le i'-1}$ as defined in Lemma~\ref{lemma:sequence}, where $n_1=m=n_{I_{s_3+1}}$ and $w':=68nw(I_s)$. Namely, if $w(I_{s_3+i})/N_{I_{s_3+i}}\ge w'$, then $68nw(I_s)\le w(I_{s_3+i+1})=w(I_{s_3+i})/N_{s_3+i}$ and $n_{I_{s_3+i+1}}=1+n_{I_{s_3+i}}$, whereas, for $w(I_{s_3+i})/N_{I_{s_3+i}}< w'$, we have $w(I_{s_3+i+1})=w(I_{s_3+i})/2$ and 
$n_{I_{s_3+i+1}}=\max(n_{I_{s_3+i}}-1,1)$. It follows that $i'$ is bounded by $8(n_1+\log\log\max(4,w(I_{s_3+1})/w'))=O(\log n+\log \tau)$. Hence, there exists an $s_4=s_3+i'+1=O(\log n+\log \tau)$ with $w(I_j)<68nw(I_s)$ for all $j\ge s_4$. Finally, this shows that $s$ is upper bounded by $s_4+\log(68n)=O(\log n+\log \tau)$.
\end{proof}

Combining the latter theorem and (\ref{treesize1}) immediately yields the following result on the size of the induced recursion tree:

\begin{theorem}
For a polynomial $f$ of degree $n$ with integer coefficients of bitsize $\tau$, the algorithm $\dcn$ induces a recursion tree $T_{\dcn}$ of size
\[
|T_{\dcn}|=\var(f,I_0)\cdot O(\log n +\log\tau)=O(n\cdot(\log n+\log\tau)),
\]
where $I_0:=(-2^{\tau+1},2^{\tau+1})$ denotes the initial interval known to contain all real roots of $f$.
\end{theorem}

\subsection{Bit Complexity Analysis}\label{sec:bitcomplexity}

We now derive an upper bound for the number of bit operations that are needed to determine isolating intervals for the real roots of $f$. We will show that, in each iteration, the costs are dominated (up to a constant factor) by the computation of  the polynomial $f_I=(x+1)^n\cdot f((ax+b)/(x+1))$ as defined in (\ref{poly:fI}). The costs for this step mainly depend on the bitsize of the endpoints, and thus on the length of the interval $I=(a,b)$. The following Lemma provides a lower bound on the width of the special nodes $J_i$, and thus also for the nodes $I\in T_i$, in terms of the separations of the roots $z_1,\ldots,z_n$:

\begin{lemma}\label{boundIj}
For a polynomial $f$ as defined in (\ref{polyf}), we can order the roots of $z_1,\ldots,z_n$ of $f$ in way such that
\[
w(J_i)>\frac{\sigma(z_i)}{4}\cdot n^{-5-2\log n},
\]
where $i=1,\ldots,n'$ and $J_1,\ldots,J_{n'}$ denote the special nodes in the recursion tree $T_{\dcn}$. In addition, the endpoints of an arbitrary interval $I\in T(i)\cup \{J_i\}$ are dyadic numbers which can be represented by $O(\tau+\log\sigma_i^{-1}+\log^2n)$ bits.
\end{lemma}

\begin{proof}
We first order the roots of $f$ with respect to their separations, 
that is, $\sigma(z_1)\ge\cdots\ge\sigma(z_n)$. For a fixed 
$\sigma\in\R^+$, let $k$ be defined such that $\sigma(z_{n-k+1})< 
\sigma\le\sigma(z_{n-k})$, that is, exactly the $k$ roots $z_{n-k+1},\ldots,z_n$ have separation less than $\sigma$. We further denote 
$I_1:=J_{i_1},\ldots,I_m:=J_{i_m}$ the special nodes such that
$$w(I_l)<w_{\min}:=\frac{\sigma}{4}\cdot n^{-5-2\log n}\text{ for all }l=1,\ldots,m,$$ 
and each 
special node $J_i\in T_{\dcn}$ which contains $I_l$ has width $w(J_i)\ge w_{\min}$. In addition, we denote $v_l:=\var(f,I_l)\ge 2$ the number of sign variations that 
we count for $I_l$. Since the intervals $I_l$ are disjoint, we have 
$v_1+\cdots+v_m\le n$, and thus $m\le n/2$. According to 
Theorem~\ref{Obreshkoff}, $v_l$ is a lower bound for the number of roots 
within the Obreshkoff area $A_n^{(l)}$ for $I_l$. Furthermore, in 
Lemma~\ref{lem:lensarea}, we have proven that any two points within $A_n^{(l)}$ 
have distance less than $4nw(I_l)<\sigma$, and thus each root contained 
in $A_n^{(l)}$ must be one of the $k$ roots $z_{n-k+1},\ldots,z_n$. Let $S_l$ be the set of all roots which are contained in $A_n^{(l)}$. 

Let us first consider the case, where the Obreshkoff areas $A_n^{(l)}$ 
are pairwise disjoint. Then, the subsets $S_l\subset A_n^{(l)}$ are also pairwise disjoint, and thus 
$$v_1+\cdots+v_s\le |S_1|+\cdots+|S_m|\le k.$$
We now turn to the more general situation, where some of the $A_n^{(l)}$ may overlap. If this happens, 
then it is possible that, in total, some of the roots contained in these areas are counted more than once. Hence, we propose the following approach: In a first step (see the subsequent construction), starting with the intervals $I_1,\ldots,I_m$, we iteratively merge intervals whose corresponding Obreshkoff areas overlap until we finally obtain intervals $I_1',\ldots,I'_{m'}$ such that
\begin{itemize} 
\item the intervals $I_1,\ldots,I_m$ are covered by $I_1',\ldots,I'_{m'}$, 
\item each $I_l'$ has width $w(I_l')<w_{\min}\cdot n^{4+2\log n}$, and
\item the Obreshkoff areas $A_n$ for $I_l'$ are pairwise disjoint.
\end{itemize}
Then, using the same argument as above yields $v_1+\cdots+v_{m}\le 
v':=\sum_{l=1}^{m'}\var(f,I'_{l})\le k$, where the latter inequality follows from the fact that the Obreshkoff area $A_n$ for each $I_{l}'$ contains only roots 
with separation less than $4nw(I'_{l})<4nw_{\min}\cdot n^{4+2\log n}<\sigma(z_{n-k})$.
It remains to show how to construct the intervals $I_l'$: We start with a list 
$\mathcal{A}$ of active intervals, where we initially set 
$\mathcal{A}:=\{I_1,\ldots,I_m\}$. In each iteration, we pick two intervals 
$I=(a,b)$ and $I'=(a',b')$ from $\mathcal{A}$ whose corresponding Obreshkoff 
areas overlap. Then, we remove $I$, $I'$, and all intervals $J\in\mathcal{A}$ 
in between $I$ and $I'$. Finally, we add the smallest interval $K$ into 
$\mathcal{A}$ which contains $I$ and $J$ (i.e., $K=(\min(a,a'),\max(b,b'))$).
We proceed in this way until we obtain intervals $I_1',\ldots,I'_{m'}$ such that the corresponding Obreshkoff areas do not overlap. It remains to show that the length of each of the so-obtained intervals is bounded by $w_{\min}\cdot n^{4+2\log n}$: At any stage, an interval $J$ obtained in the above iteration, covers a certain number $s$ of intervals from $I_1,\ldots,I_m$. By induction on $s$, we prove that 
\begin{align}
w(J)<w_{\min}\cdot s^4n^{2\log s},\text{ and thus }w(J)\le w_{\min}\cdot n^{2\log n+4}.\label{widthJ}
\end{align}
If $s=1$, then trivially $J=I_l$ for some $l$ which proves the claim for $s=1$. An interval $J$ covering $s+1$ intervals from $I_1,\ldots,I_m$ is obtained by merging two intervals $I$ and $I'$ which cover $s_1$ and $s_2$ intervals, respectively, where $s_1+s_2\le s+1$ and, w.l.o.g., $1\le s_1\le s_2$. From Lemma~\ref{lem:lensarea}, it follows that the distance between $I$ and $I'$ is bounded by $4n^2\cdot\min(w(I'),w(I))$ because the corresponding Obreshkoff areas overlap. Hence, $J$ has width  
\begin{align*}
w(J)&\le w(I)+w(I')+4n^2\cdot\min(w(I'),w(I))<w_{\min}\cdot(s_2^4n^{2\log s_2}+(4n^2+1)s_1^4n^{2\log s_1})\\
&<w_{\min}\cdot(s_2^4n^{2\log s_2}+8s_1^4n^{2\log s_1+2})=w_{\min}\cdot(s_2^4n^{2\log s_2}+8s_1^4n^{2\log 2s_1})
\end{align*}
If $2s_1\ge s_2$, then $$s_2^4n^{2\log s_2}+8s_1^4n^{2\log 2s_1}\le (8s_1^4+s_2^4)n^{2\log 2s_1}\le (s_1+s_2)^4n^{2\log(s_1+s_2)}.$$ Otherwise, we have $$s_2^4n^{2\log s_2}+8s_1^4n^{2\log 2s_1}\le (8s_1^4+s_2^4)n^{2\log s_2}< (s_1+s_2)^4n^{2\log(s_1+s_2)},$$
and thus (\ref{widthJ}) follows. 

Hence, we count at most $k$ sign variations for the intervals $I_l$, in total. Now, the same argument as used in Section~\ref{sec:idea} to show that there exists at most $n'=\var(f,[-2^{\tau+1},2^{\tau+1}])$ special nodes also shows that the number of all special nodes $J$ with $w(J)\le\sigma$ is bounded by $k$. Namely, we start with special nodes $I_1,\ldots,I_m$ with $\sum_{l=1}^m\var(f,I_l)\le k$, and whenever a special node is subdivided, the value $\mu:=\sum_I\var(f,I)-\#\{I:\var(f,I)\ge 2\}$ decreases by at least one. For the result on the width of the intervals $I_j$, we consider integers $k_1,\ldots,k_s$, with $k_1<k_2<\cdots<k_s$, such that $$\sigma(z_1)=\cdots=\sigma(z_{n-k_s})>\sigma(z_{n-k_s+1})=\cdots=\sigma(z_{n-k_{s-1}})>\sigma(z_{n-k_1+1})=\cdots=\sigma(z_n).$$
There exist at most $k_j$ special nodes $J_i$ with $w(J_i)<\frac{\sigma(z_{n-k_j})}{4}n^{-4-2\log n}$, and there cannot exist a special node $J_i$ with $w(J_i)<\sigma_f=\sigma(z_n)$. In addition, the number of special nodes with $w(J_i)<\sigma(z_{n-k_1})n^{-4-2\log n}/4$ is bounded by $k_1$. Then, by induction, it follows that $w(J_i)<\sigma(z_{i})n^{-4-2\log n}/4$ because the $J_i$ are ordered with respect to their length.

For the bit complexity of the endpoints of an interval $I=(a,b)\in T(i)\cap \{J_i\}$, we remark that, due to our construction, $a$ and $b$ are both dyadic numbers with modulus bounded by $2^{\tau+1}$. In addition, the denominator of $a$ and $b$ in its dyadic representation is bounded by $\max(1,\log w(I)^{-1})$. Hence, we can represent both endpoints of $I$ with $O(\tau+\log\sigma_i^{-1}+\log^2n)$ many bits.
\end{proof}

We will now derive our final result on the bit complexity of $\dcn$: In each step of the algorithm, we have to compute the polynomial $f_I(x)=(x+1)^n\cdot f((ax+b)/(x+1))$, where $I=(a,b)$ is the interval that is actually processed. The latter computation decomposes into computing $f_I^*:=f(a+(b-a)x)$, reversing the coefficients, and then applying a Taylor shift by $1$ (i.e., $x\mapsto x+1$). For  the computation of $f_I^*$, we first shift $f$ by $a$, and then scale by a factor $b-a=w(I)$ which is a power of two. Using asymptotically fast Taylor shift~\cite{Ger04,GG97}, the shift $x\mapsto x+a$ (i.e., the computation of $f(x+a)$) demands for $\tilde{O}(n^2(\tau+\log w(I)^{-1}))$ bit operations. The scaling $x\mapsto (b-a)\cdot x$ is achieved by just shifting the $i$-th coefficient of $f(x+a)$ by $i\cdot\log(b-a)^{-1}=i\cdot\log w(I)^{-1}$ many bits. Then, the resulting polynomial $f_I^*$ has coefficients of modulus $2^{O(n\tau)}$, and the denominators of their dyadic representations are bounded by $2^{O(n(\tau+\log w(I)^{-1}))}$. Hence, reversing the coefficients of $f^*_I$, and then applying a Taylor shift by $1$, demands for $\tilde{O}(n^2(\tau+\log w(I)^{-1}))$ bit operations. In summary, the costs for computing $f_I$ are bounded by $\tilde{O}(n^2(\tau+\log w(I)^{-1}))$. The same bound further applies to the computation of $\lambda_1$ and $\lambda_2$ in Step 3 (b), (\ref{def:lambda}) of our algorithm because, in this step, we have to evaluate a polynomial of degree $n$ and bitsize $\tau$ at a $(\tau+\log w(I)^{-1})$-bit number. If $I$ is non-terminal, then we also have to compute the number of sign variations for the intervals  $B_1$, $B_2$, $I_1'$ and $I_2'$. The same argument as above also shows that we can do so using $\tilde{O}(n^2(\tau+\log N_I+\log w(I)^{-1}))$ bit operations. 

For an interval $I\in T(i)\cup J_i$, we have $$\log N_I+\log w(I)^{-1}=O(\log N_I +\log w(I_j)^{-1})=O(\tau+\log\sigma(z_i)^{-1}+\log^2 n));$$
see Lemma~\ref{lemma:bounds1} and~\ref{boundIj}. Thus, the total costs for all computations at $I$ are bounded by $\tilde{O}(n^2(\tau+\sigma(z_i)^{-1}))$ bit operations. It remains to consider a terminal interval $I$ which is one of the two children of a special node $J_i$. In this case, we only have to bound the cost for the computation of $f_I$ because $\var(f,I)\le 1$. 
Since $w(I)=w(J_i)/2$, the latter computation needs $\tilde{O}(n^2(\tau+\log w(I_j)^{-1}))=\tilde{O}(n^2(\tau+\log\sigma(z_i)^{-1}))$ bit operations as well. 

In Theorem~\ref{thm:main1}, we have shown that $|T_i|=O(\log n+\log\tau)$ for all $i$, and thus the total costs for isolating the real roots of $f$ are bounded by
\[
O(\log n+\log\tau)\cdot\sum_{i=1}^{n'}\tilde{O}(n^2(\tau+\log\sigma(z_i)^{-1}))=\tilde{O}(n^3\tau)
\]
since $\sum_{i=1}^{n'}\log \sigma(z_i)^{-1}=O(n\tau+\sum_{i=1}^{n}\log \sigma(z_i)^{-1})=\tilde{O}(n\tau)$ according to Lemma 19 in~\cite{sagraloff-complexity}. We fix this result:

\begin{theorem}
For a polynomial $f$ of degree $n$ with integer coefficients of modulus less than $2^\tau$, $\dcn$ isolates the real roots of $f$ using no more than $\tilde{O}(n^3\tau)$ bit operations.
\end{theorem} 

\section{Conclusion}

We introduced the first subdivision method to isolate the real roots of a polynomial which achieves the record bound $\tilde{O}(n^3\tau)$ for the bit complexity of this fundamental problem. In comparison to the asymptotically fast algorithms by Pan and Sch\"onhage from the 80tes which compute all complex roots, the new approach is much simpler and can be considered practical to an extremely high degree. The algorithm is based on a novel subdivision technique which combines Descartes' Rule of Signs and Newton iteration. As a consequence, our algorithm shows quadratic convergence towards the roots in most steps. 

So far, the approach only applies to polynomials with integer (or rational) coefficients. In~\cite{s-bitstream-mcs11}, we showed how to modify a subdivision algorithm which uses exact computation in each step such that it also applies to polynomials $f$ with arbitrary real coefficients that can be approximated to any specified error bound (bitstream coefficients). Following this approach, it seems reasonable to express the bit complexity in terms of the geometry of the roots $z_1,\ldots,z_n$ of $f$. For instance, the modified version of the Descartes method from~\cite{sagraloff-complexity} isolates all real roots using $\tilde{O}(n(n\Gamma+\Sigma)^2)$ bit operations, where $\Gamma$ constitutes a bound on the logarithm of the modulus of the roots and $\Sigma:=\sum_{i=1}^n \log\sigma(z_i)^{-1}$. Due to the quadratic convergence achieved by the new algorithm, we expect a corresponding bitstream version to perform the same task by a considerably lower number of bit operations.

\include{paperbbl}
%\bibliography{localref,ref,yap,exact,alge}
%\bibliographystyle{abbrv}
\newpage
\section{Appendix}\label{sec:appendix}

\begin{algorithm}
\caption{$\textsc{Dsc}^2$}
\label{alg:dcm}
\begin{algorithmic}
\REQUIRE {polynomial $f =\sum_{0 \le i \le n} a_i x^i\in\Z[x]$ with integer coefficients $a_i$, $|a_i|<2^\tau$ for all $i$.}\medskip
\ENSURE{returns a list $\mathcal{O}$ of disjoint isolating intervals for all real roots of $f$}

\STATE $I_0 \assign (-2^{\tau+1},2^{\tau+1})$; $N_{I_0}\assign 4$
\STATE $\mathcal{A} \assign \sset{(I_0,N_{I_0})}$; $\mathcal{O} \assign \emptyset$ \hfill \COMMENT{\emph{list of
active and isolating intervals}}
%\WHILE {({\bf true})}  
\REPEAT
\STATE $(I,N_I)$ some element in $\mathcal{A}$ with $I=(a,b)$; 
delete $(I,f_I)$ from $\mathcal{A}$
\STATE $v\assign\var(f,I)$
\IF {$v=0$}
\STATE do nothing\hfill\COMMENT{\emph{$I$ contains no root}}
\ELSIF {$v=1$}
\STATE add $I$ to $\mathcal{O}$\hfill\COMMENT{\emph{$I$ isolates a real root}}
%\IF {$v\le 1$}
%     \IF {$v=0$}
%         \STATE do nothing\hfill\COMMENT{$I$ contains no root}
%     \ELSE
%         \STATE add $I$ to $\mathcal{O}$\hfill\COMMENT{$I$ isolates a real root}
%     \ENDIF
\ELSIF {$v>1$}
    \STATE $B_1\assign (a,a+\frac{w(I)}{N_I})$; $B_2\assign (b-\frac{w(I)}{N_I},b)$
    \IF {$\var(f,B_1)=v$ \OR $\var(f,B_2)=v$}
        \STATE for the unique $i\in\{1,2\}$ with $\var(f,B_i)=v$, add $(B_i,N_{B_i})\assign (B_i,N_I^2)$ to $\mathcal{A}$\\ \vspace{0.10cm}
        \hfill\COMMENT{\emph{$B_i$ contains all real roots within $I$} \vspace{0.10cm}}
    \ELSE 
        %\STATE $I^*=(a^*,b^*)\assign \left(a+\frac{w(I)}{3N_I},b-\frac{w(I)}{3N_I}\right)$
        \STATE $\lambda_1\assign a-v\cdot\frac{f(a)}{f'(a)}$; $\lambda_2\assign b-v\cdot\frac{f(b)}{f'(b)}$\\ 
        %\hfill\quad\COMMENT{\emph{If $\xi:=a-v\cdot\frac{f(a)}{f'(a)}\in I^*$, $\lambda_1=\xi$; otherwise, we round to the nearest endpoint of $I^*$}\vspace{0.10cm}}
\STATE for $i=1,2$:\\
$k_i\assign\min(\max(\lfloor  4N_I\cdot\frac{\lambda_i-a}{b-a}\rfloor,2),4N_I-2)$; $I_{i}'\assign (a+(k_i-2)\cdot\frac{w(I)}{4N_I},a+(k_i+2)\cdot\frac{w(I)}{4N_I})$;\\  \vspace{0.10cm} 
\STATE \quad\emph{\{$m_k:=a+k\cdot\frac{w(I)}{4N_I}$ is the $k$-th subdivision point when decomposing $I$ into $4N_I$ equally}
\STATE\quad\emph{sized intervals; $m_{k_i}$ is closest to $\lambda_i$; $I'_i$ has width $\frac{w(I)}{N_I}$ and is centered at $m_{k_i}$.}\}\vspace{0.10cm}
%\STATE for $i=1,2$: $I_i'\assign \left(a+k_i\cdot\frac{w(I)}{4N_I},a+(k_i+1)\cdot\frac{w(I)}{4N_I}\right)$\\  \vspace{0.10cm} \hfill\COMMENT{\emph{$I_i'$ is an interval of size $w(I)/N_I$ centered at $a+k_i\cdot\frac{w(I)}{4N_I}$}\vspace{0.10cm}}
    \IF {$\var(f,I_1')=v$ \OR $\var(f,I_2')=v$}
         \STATE choose the smallest $i\in\{1,2\}$ with $\var(f,I_i')=v$ and add $(I_i',N_{I_i'})\assign(I_i',N_I^2)$ to $\mathcal{A}$\\   \vspace{0.10cm}
         \hfill\COMMENT{\emph{If $\var(f,I_i')=v$, then $I_i'$ contains all roots within $I$}\vspace{0.10cm}}
    \ELSE
         \STATE add $(I_1,N_{I_1}):=((a,m(I)),\max(4,\sqrt{N_I}))$, $(I_2,N_{I_2}):=((m(I),b),\max(4,\sqrt{N_I}))$ to $\mathcal{A}$
         \IF {$f(m(I))=0$}
             \STATE add $[m(I),m(I)]$ to $\mathcal{O}$
         \ENDIF
   \ENDIF
   \ENDIF
\ENDIF
\UNTIL{$\mathcal{A}$ is empty}
\RETURN $\mathcal{O}$
\end{algorithmic}
\end{algorithm}

\end{document}

%% file: papermacros.tex
\input bookmacros.tex

\newlength{\proofpostskipamount}\newlength{\proofpreskipamount}
               {\par\vspace{0.5ex}\noindent{\bf Proof #1:}\hspace{0.5em}}%
               {\nopagebreak%
                \strut\nopagebreak%
                \hspace{\fill}\qed\par\medskip\noindent}

\newlength{\mydefwidth}
\newlength{\mytextwidth}

\newcommand{\myurl}[1]{{\footnotesize \url{#1}}}

%% file: bookmacros.tex
\newcommand{\donotshow}[1]{}

\newcommand{\ignore}[1]{}

\newtheorem{theorem}{Theorem}
\newtheorem{lemma}[theorem]{Lemma}

\newtheorem{definition}{Definition}
\providecommand{\qed}{\rule[-0.2ex]{0.3em}{1.4ex}}

\newcommand{\assign}{\mathbin{:=}}

\providecommand{\R}{\mathbb{R}}
\newcommand{\N}{\mathbb{N}}
\newcommand{\Z}{\mathbb{Z}}

\newcommand{\C}{\mathbb{C}}

\newlength{\setspacing}
\providecommand{\sset}[1]{\{\, #1 \,\}}

\newcommand{\mbegin}{\{\ \ }
\newcommand{\mend}{\}}

\newlength{\mleftindent}
\newlength{\mindent}
\newlength{\mboxwidth}
\newcommand{\mincrement}{\addtolength{\mboxwidth}{-\mindent}}
\newcommand{\mdecrement}{\addtolength{\mboxwidth}{\mindent}}

\newlength{\preprogramskip}
\newlength{\postprogramskip}
\newlength{\mexpwidth}
\newlength{\mexpindent}
\newcommand{\indentafterkeyword}{\hspace*{0.5em}}

\newcommand{\mslifelse}[3]  %if is short else is long
{\setlength{\mexpwidth}{\mboxwidth}%
\settowidth{\mexpindent}{{\bf if\indentafterkeyword}}%
\addtolength{\mexpwidth}{-\mexpindent}%
{\bf if\indentafterkeyword}\parbox[t]{\mexpwidth}{#1}\\
\mincrement \mbegin \parbox[t]{\mboxwidth}{#2 \mend} \mdecrement \\
{\bf else} \\
\mincrement \mbegin \parbox[t]{\mboxwidth}{#3}\\
\mend \mdecrement
}

{\vspace{-0.3em}\begin{itemize}
\setlength{\itemsep}{0.2\itemsep}%
\setlength{\parskip}{0.2\parskip}%
\setlength{\topsep}{0.0\topsep}%
}
{\end{itemize}}

{\begin{itemize}
\setlength{\itemsep}{0.2\itemsep}%
\setlength{\parskip}{0.2\parskip}%
\setlength{\topsep}{0.0\topsep}%
}
{\end{itemize}}

{\vspace{-0.3em}\begin{description}
\setlength{\itemsep}{0.2\itemsep}%
\setlength{\parskip}{0.2\parskip}%
\setlength{\topsep}{0.0\topsep}%
}
{\end{description}}

{\begin{description}
\setlength{\itemsep}{0.2\itemsep}%
\setlength{\parskip}{0.2\parskip}%
\setlength{\topsep}{0.0\topsep}%
}
{\end{description}}

{\vspace{-0.3em}\begin{enumerate}
\setlength{\itemsep}{0.2\itemsep}%
\setlength{\parskip}{0.2\parskip}%
\setlength{\topsep}{0.0\topsep}%
}
{\end{enumerate}}

{\begin{enumerate}
\setlength{\itemsep}{0.2\itemsep}%
\setlength{\parskip}{0.2\parskip}%
\setlength{\topsep}{0.0\topsep}%
}
{\end{enumerate}}